\documentclass[onecolumn]{revtex4-2}
\usepackage[
	bookmarks  = false,
	colorlinks = true,
	linkcolor  = blue,
	urlcolor   = blue,
	citecolor  = blue]{hyperref}
\usepackage{amsmath,amssymb,amsthm,braket}
\usepackage{graphicx,bm,bbm}
\usepackage[T1]{fontenc}
\newtheorem{theorem}{Theorem}

\newtheorem{corollary}{Corollary}
\newtheorem{definition}{Definition}
\newtheorem{proposition}{Proposition}
\newtheorem{example}{Example}
\makeatletter
\def\thmhead@plain#1#2#3{%
  \thmname{#1}\thmnumber{\@ifnotempty{#1}{ }\@upn{#2}}%
  \thmnote{ {\the\thm@notefont#3}}}
\let\thmhead\thmhead@plain
\makeatother
\DeclareMathOperator{\Tr}{Tr}

\begin{document}
%%%%%%%%%%%%%%%%%%%%%%
%       Title        %
%%%%%%%%%%%%%%%%%%%%%%
\title{Factorizability of multi-party quantum sequence discrimination \\
under local operations and classical communication}
\author{Donghoon Ha}
\affiliation{Department of Applied Mathematics and Institute of Natural Sciences, Kyung Hee University, Yongin 17104, Republic of Korea}
\author{Jeong San Kim}
\email{freddie1@khu.ac.kr}
\affiliation{Department of Applied Mathematics and Institute of Natural Sciences, Kyung Hee University, Yongin 17104, Republic of Korea}
%%%%%%%%%%%%%%%%%%%%%%
%      Abstract      %
%%%%%%%%%%%%%%%%%%%%%%
\begin{abstract}
We consider multi-party quantum sequence discrimination under local operations and classical communication(LOCC), and provide conditions under which the optimal LOCC discrimination of a multi-party quantum sequence ensemble can be factorized into that of each individual ensemble. In other words, the optimal LOCC discrimination of a multi-party quantum sequence ensemble can be achieved just by performing optimal LOCC discrimination independently at each step of the quantum sequence. We also illustrate through examples of multi-party quantum states that such factorizability of optimal LOCC discrimination is possible. We further establish a necessary and sufficient condition under which the optimal LOCC discrimination of a multi-party quantum state ensemble can be realized just by guessing the state with the largest probability. Our results can provide a useful application to investigate the fundamental limits of quantum data hiding.
\end{abstract}
%%%%%%%%%%%%%%%%%%%%%%
\maketitle
%%%%%%%%%%%%%%%%%%%%%%

%%%%%%%%%%%%%%%%%%%%%%
%      Section       %
%%%%%%%%%%%%%%%%%%%%%%
\section{Introduction}\label{sec:intro}
%%%%%%%%%%%%%%%%%%%%%%
One of the fundamental tasks in quantum information processing is to determine which state is prepared from a quantum state ensemble, that is, quantum state discrimination\cite{chef2000,berg2007,barn20091,bae2015}.
In general, orthogonal quantum states can be perfectly discriminated using an appropriate measurement, whereas no measurement can perfectly discriminate non-orthogonal quantum states.
To address this limitation, various state discrimination strategies have been developed to optimally discriminate nonorthogonal quantum states, including minimum-error discrimination, unambiguous discrimination and maximum-confidence discrimination\cite{hels1969,ivan1987,diek1988,pere1988,crok2006}.
%%%%%%%%%%%%%%%%%%%%%%

%%%%%%%%%%%%%%%%%%%%%%
In multi-party systems, the situation is more complicated due to the possible constraints of \emph{local operations and classical communication}(LOCC); some orthogonal quantum states cannot be perfectly discriminated only by LOCC\cite{chit2014,benn19991,ghos2001}. 
Moreover, there exist some multi-party non-orthogonal quantum states that cannot be optimally discriminated using only LOCC\cite{pere1991,duan2007,chit2013}. 
Thus it is important and even necessary to investigate the fundamental limit of optimal LOCC discrimination for the full characterization of multi-party quantum state discrimination.
However, characterizing optimal LOCC discrimination remains a hard task, largely due to the lack of a well-established mathematical structure for LOCC.
%%%%%%%%%%%%%%%%%%%%%%

%%%%%%%%%%%%%%%%%%%%%%
When multiple quantum state ensembles are involved, the situation becomes significantly more intricate due to the need to determine not just a single quantum state, but a sequence of states, each independently prepared from a quantum state ensemble. 
Quantum sequence discrimination naturally appears in various quantum information tasks such as multiple-copy discrimination and quantum change-point detection\cite{higg2011,sent2016,sent2017}. 
Accordingly, a deeper understanding of quantum sequence discrimination is beneficial to uncovering broader principles in quantum information processing.
%%%%%%%%%%%%%%%%%%%%%%

%%%%%%%%%%%%%%%%%%%%%%
Recently, it was shown that the minimum-error discrimination of a quantum sequence ensemble can always be factorized into that of each individual ensemble. In other words, the minimum-error discrimination of a quantum sequence ensemble can be realized just by performing minimum-error discrimination independently at each step of the quantum sequence.
This factorizability was also proven to hold in optimal unambiguous state discrimination\cite{gupt20241,gupt20242}.
This naturally raises the question of whether the optimal LOCC discrimination of a multi-party quantum sequence ensemble can also be factorized into that of each individual ensemble.
%%%%%%%%%%%%%%%%%%%%%%

%%%%%%%%%%%%%%%%%%%%%%
Studying LOCC discrimination of a multi-party quantum state ensemble plays an central role in various quantum information processing tasks such as quantum data hiding and quantum secret sharing\cite{terh2001,divi2002,egge2002,raha2015,wang2017}.
The nonlocality revealed through the gap between global and LOCC-based discrimination strategies is a key resource to prevent individual parties from accessing concealed information unless they collaborate.
This inherent limitation of LOCC protocols underscores their role not merely as a technical constraint, but as an essential feature enabling quantum privacy, security, and control.
%%%%%%%%%%%%%%%%%%%%%%

%%%%%%%%%%%%%%%%%%%%%%
Here, we consider multi-party quantum sequence discrimination under LOCC constraints, and provide conditions under which the optimal LOCC discrimination of a multi-party quantum sequence ensemble can be factorized into that of each individual ensemble. 
In other words, the optimal LOCC discrimination of a multi-party quantum sequence ensemble can be achieved just by performing optimal LOCC discrimination independently at each step of the quantum sequence. 
Moreover, our results are illustrated with examples of multi-party quantum states in an arbitrary dimension, showing the cases that such factorizability of optimal LOCC discrimination is possible.
%%%%%%%%%%%%%%%%%%%%%%

%%%%%%%%%%%%%%%%%%%%%%
This paper is organized as follows.
In Sec.~\ref{sec:qsd}, we first recall the definitions and properties of multi-party quantum state discrimination under LOCC constraints.
We further provide a necessary and sufficient condition under which the optimal LOCC discrimination of a multi-party quantum state ensemble can be realized just by guessing the state with the largest probability.
In Sec.~\ref{sec:qseq}, we review the definitions and properties of quantum sequence discrimination. 
In Sec.~\ref{sec:ldmqs}, we establish conditions under which the optimal LOCC discrimination of a multi-party quantum sequence ensemble can be factorized into that of each individual ensemble.
Our results are illustrated with examples of multi-party quantum states in an arbitrary dimension.
In Sec.~\ref{sec:dis}, we summarize our results and discuss a useful application.
We also suggest a direction for future research.
%%%%%%%%%%%%%%%%%%%%%%

%%%%%%%%%%%%%%%%%%%%%%
%      Section       %
%%%%%%%%%%%%%%%%%%%%%%
\section{Quantum state discrimination}\label{sec:qsd}
%%%%%%%%%%%%%%%%%%%%%%
For a multi-party Hilbert space $\mathcal{H}=\bigotimes_{k=1}^{m}\mathbb{C}^{d_{k}}$ with $m,d_{1},\ldots,d_{m}\geqslant2$, let us denote by $\mathbb{H}$ the set of all Hermitian operators acting on $\mathcal{H}$.
We also denote by $\mathbb{H}_{+}$ the set of all positive-semidefinite operators in $\mathbb{H}$, that is,
\begin{equation}\label{eq:hpdf}
\mathbb{H}_{+}=\{E\in\mathbb{H}\,|\,E:\mbox{positive semidefinite}\}.
\end{equation}
A multi-party quantum state is represented by $\rho\in\mathbb{H}_{+}$ with $\Tr\rho=1$.
A measurement is described by $\{M_{i}\}_{i}\subseteq\mathbb{H}_{+}$ satisfying $\sum_{i}M_{i}=\mathbbm{1}$ where $\mathbbm{1}$ is the identity operator in $\mathbb{H}$.
When a measurement $\{M_{i}\}_{i}$ is applied to a state $\rho$, the probability of observing the measurement outcome corresponding to $M_{i}$ is $\Tr(\rho M_{i})$.
%%%%%%%%%%%%%%%%%%%%%%

%%%%%%%%%%%%%%%%%%%%%%
Let us consider the situation of discriminating the multi-party quantum states from the ensemble
\begin{equation}\label{eq:esb}
\mathcal{E}=\{\eta_{i},\rho_{i}\}_{i=1}^{n},
\end{equation}
where the state $\rho_{i}$ is prepared with the \emph{nonzero} probability $\eta_{i}$ for each $i=1,\ldots,n$.
To guess the prepared state from the ensemble $\mathcal{E}$, we implement the decision rule using a measurement 
\begin{equation}\label{eq:mes}
\mathcal{M} = \{M_{i}\}_{i=1}^{n}, 
\end{equation}
where detecting $M_{i}$ indicates that the prepared state is inferred to be $\rho_{i}$.
The \emph{minimum-error discrimination} of $\mathcal{E}$ is to achieve the maximum average probability of correctly guessing the prepared state from $\mathcal{E}$, that is,
\begin{equation}\label{eq:pge}
p_{\sf G}(\mathcal{E})=\max_{\mathcal{M}}\sum_{i=1}^{n}\eta_{i}\Tr(\rho_{i}M_{i}),
\end{equation}
where the maximum is taken over all possible measurements\cite{hels1969}.
%%%%%%%%%%%%%%%%%%%%%%

%%%%%%%%%%%%%%%%%%%%%%
\begin{definition}\label{def:sep}
$E\in\mathbb{H}_{+}$ is called \emph{separable} if it can be expressed as the sum of positive-semidefinite product operators, that is,
\begin{equation}\label{eq:sep}
E=\sum_{s}\bigotimes_{k=1}^{m}E_{s,k},
\end{equation}
where $E_{s,k}$ is a positive-semidefinite operator on $\mathbb{C}^{d_{k}}$ of $\mathcal{H}$ for each $k=1,\ldots,m$.
\end{definition}
%%%%%%%%%%%%%%%%%%%%%%
\noindent
We denote by $\mathbb{SEP}$ the set of all separable operators in $\mathbb{H}_{+}$, that is,
\begin{equation}\label{eq:sepd}
\mathbb{SEP}=\{E\in\mathbb{H}_{+}\,|\,E:\mbox{separable}\}.
\end{equation}
A measurement $\{M_{i}\}_{i}$ is called a \emph{separable measurement} if $\{M_{i}\}_{i}\subseteq\mathbb{SEP}$. We also say that a measurement is a \emph{LOCC measurement} if it can be realized by LOCC. Note that every LOCC measurement is a separable measurement\cite{chit2014}.
%%%%%%%%%%%%%%%%%%%%%%

%%%%%%%%%%%%%%%%%%%%%%
When only separable measurements are available, we denote by $p_{\sf SEP}(\mathcal{E})$ the maximum average probability of correctly guessing the prepared state from the ensemble $\mathcal{E}$ in Eq.~\eqref{eq:esb}, that is,
\begin{equation}\label{eq:psep}
p_{\sf SEP}(\mathcal{E})=\max_{\textsf{Separable}\,\mathcal{M}}\sum_{i=1}^{n}\eta_{i}\Tr(\rho_{i}M_{i}).
\end{equation}
Similarly, we denote
\begin{equation}\label{eq:ple}
p_{\sf L}(\mathcal{E})=\max_{\textsf{LOCC}\,\mathcal{M}}\sum_{i=1}^{n}\eta_{i}\Tr(\rho_{i}M_{i}),
\end{equation}
where the maximum is taken over all possible LOCC measurements.
%%%%%%%%%%%%%%%%%%%%%%

%%%%%%%%%%%%%%%%%%%%%%
\begin{definition}\label{def:bpo}
$E\in\mathbb{H}$ is called \emph{block positive} if it has non-negative mean value for all separable states, that is,
\begin{equation}\label{eq:bpo}
\Tr(E\sigma)\geqslant0
\end{equation}
for any separable state $\sigma$.
\end{definition}
%%%%%%%%%%%%%%%%%%%%%%
\noindent
We use $\mathbb{SEP}^{*}$ to denote the set of all block-positive operators in $\mathbb{H}$, that is,
\begin{equation}\label{eq:bpdf}
\mathbb{SEP}^{*}=\{E\in\mathbb{H}\,|\,E:\mbox{block positive}\}.
\end{equation}
Note that $\mathbb{SEP}^{*}$ is the dual cone of $\mathbb{SEP}$, but $\mathbb{SEP}$ is also the dual cone of $\mathbb{SEP}^{*}$ because $\mathbb{SEP}$ is convex and closed\cite{boyd2004}.
We also note that
\begin{equation}\label{eq:subs}
\mathbb{SEP}\subseteq\mathbb{H}_{+}\subseteq\mathbb{SEP}^{*}\subseteq\mathbb{H},
\end{equation}
where the first inclusion is from the definition of $\mathbb{SEP}$, the second inclusion is from the fact that $\Tr(EF)\geqslant0$ for all $E\in\mathbb{H}_{+}$ and all $F\in\mathbb{SEP}$, and
the last inclusion is from the definition of $\mathbb{SEP}^{*}$.
%%%%%%%%%%%%%%%%%%%%%%

%%%%%%%%%%%%%%%%%%%%%%
For a given ensemble $\mathcal{E}$ in Eq.~\eqref{eq:esb}, the following proposition provides a necessary and sufficient condition for a separable measurement to realize $p_{\sf SEP}(\mathcal{E})$ in Eq.~\eqref{eq:psep}.
%%%%%%%%%%%%%%%%%%%%%%
\begin{proposition}[\cite{ha20231}]\label{pro:nsps}
For a multi-party quantum state ensemble $\mathcal{E}=\{\eta_{i},\rho_{i}\}_{i=1}^{n}$ and a separable measurement $\mathcal{M}=\{M_{i}\}_{i=1}^{n}$, $\mathcal{M}$ provides $p_{\sf SEP}(\mathcal{E})$ if and only if there exists $H\in\mathbb{H}$ satisfying 
\begin{subequations}\label{eq:tmhr}
\begin{eqnarray}
H-\eta_{i}\rho_{i}\in\mathbb{SEP}^{*},\label{eq:tmhr1}\\
\Tr[M_{i}(H-\eta_{i}\rho_{i})]=0\label{eq:tmhr2}
\end{eqnarray}
\end{subequations}
for all $i=1,\ldots,n$. 
\end{proposition}
%%%%%%%%%%%%%%%%%%%%%%

%%%%%%%%%%%%%%%%%%%%%%
From the definitions of $p_{\sf G}(\mathcal{E})$, $p_{\sf SEP}(\mathcal{E})$ and $p_{\sf L}(\mathcal{E})$ in Eqs.~\eqref{eq:pge}, \eqref{eq:psep} and \eqref{eq:ple}, respectively, we have
\begin{equation}\label{eq:ineqp}
p_{\sf L}(\mathcal{E})\leqslant
p_{\sf SEP}(\mathcal{E})\leqslant
p_{\sf G}(\mathcal{E}).
\end{equation}
For each $i=1,\ldots,n$, guessing the prepared state as $\rho_{i}$ is obviously an LOCC measurement, therefore we have
\begin{equation}\label{eq:plet}
\eta_{i}\leqslant p_{\sf L}(\mathcal{E}).
\end{equation}
%%%%%%%%%%%%%%%%%%%%%%

%%%%%%%%%%%%%%%%%%%%%%
The following theorem establishes a necessary and sufficient condition for Inequality~\eqref{eq:plet} to be saturated.
%%%%%%%%%%%%%%%%%%%%%%
\begin{theorem}\label{thm:plq1}
For a multi-party quantum state ensemble $\mathcal{E}=\{\eta_{i},\rho_{i}\}_{i=1}^{n}$ and $x\in\{1,\ldots,n\}$, we have
\begin{equation}\label{eq:plq1}
p_{\sf L}(\mathcal{E})=\eta_{x}
\end{equation}
if and only if 
\begin{equation}\label{eq:bpnc}
\eta_{x}\rho_{x}-\eta_{i}\rho_{i}\in\mathbb{SEP}^{*}
\end{equation}
for all $i=1,\ldots,n$. In this case, we have
\begin{equation}\label{eq:itce}
p_{\sf L}(\mathcal{E})=p_{\sf SEP}(\mathcal{E}).
\end{equation}
\end{theorem}
%%%%%%%%%%%%%%%%%%%%%%
\begin{proof}
We first note that it is already established that 
\begin{equation}\label{eq:bpim}
p_{\sf SEP}(\mathcal{E})=\eta_{x}
\end{equation}
if and only if Condition~\eqref{eq:bpnc} is satisfied\cite{ha20231}.
%%%%%%%%%%%%%%%%%%%%%%

%%%%%%%%%%%%%%%%%%%%%%
Let us first assume that Eq.~\eqref{eq:plq1} is satisfied.
For each $i=1,\ldots,n$, we have
\begin{equation}\label{eq:etin}
\eta_{x}-\Tr[(\eta_{x}\rho_{x}-\eta_{i}\rho_{i})\sigma]=\eta_{x}\Tr[\rho_{x}(\mathbbm{1}-\sigma)]
+\eta_{i}\Tr(\rho_{i}\sigma)\leqslant p_{\sf L}(\mathcal{E})=\eta_{x}
\end{equation}
for any separable state $\sigma$, where the last equality is due to our assumption and the inequality is from the definition of $p_{\sf L}(\mathcal{E})$ in Eq.~\eqref{eq:ple} and the fact that $\{\sigma,\mathbbm{1}-\sigma\}$ is an LOCC measurement\cite{ha20242}.
Thus, Inequality~\eqref{eq:etin} leads us to
\begin{equation}\label{eq:lubc}
\Tr[(\eta_{x}\rho_{x}-\eta_{i}\rho_{i})\sigma]\geqslant0
\end{equation}
for any separable state $\sigma$, therefore Condition~\eqref{eq:bpnc} holds.
Moreover, Eq.~\eqref{eq:itce} is satisfied because Eqs.~\eqref{eq:plq1} and \eqref{eq:bpim} lead us to Eq.~\eqref{eq:itce}.
%%%%%%%%%%%%%%%%%%%%%%

%%%%%%%%%%%%%%%%%%%%%%
Conversely, we suppose that Condition~\eqref{eq:bpnc} is satisfied.
This assumption implies Eq.~\eqref{eq:bpim}.
From Eq.~\eqref{eq:bpim} together with Inequalities~\eqref{eq:ineqp} and \eqref{eq:plet}, we have
\begin{equation}\label{eq:spin}
\eta_{x}\leqslant p_{\sf L}(\mathcal{E})\leqslant p_{\sf SEP}(\mathcal{E})=\eta_{x}.
\end{equation}
Thus, Eqs.~\eqref{eq:plq1} and \eqref{eq:itce} hold.
\end{proof}
%%%%%%%%%%%%%%%%%%%%%%

%%%%%%%%%%%%%%%%%%%%%%
From Inequality~\eqref{eq:plet} together with Eq.~\eqref{eq:plq1}, we note that Theorem~\ref{thm:plq1} provides a necessary and sufficient condition under which $p_{\sf L}(\mathcal{E})$ is equal to the largest probability of $\{\eta_{i}\}_{i=1}^{n}$. 
Theorem~\ref{thm:plq1} also leads us to the following corollary.
%%%%%%%%%%%%%%%%%%%%%%
\begin{corollary}\label{cor:qdh}
For a multi-party quantum state ensembles $\mathcal{E}=\{\eta_{i},\rho_{i}\}_{i=1}^{n}$, we have
\begin{equation}\label{eq:fron}
p_{\sf L}(\mathcal{E})=\tfrac{1}{n}
\end{equation}
if and only if 
\begin{subequations}\label{eq:frcd}
\begin{gather}
\eta_{1}=\cdots=\eta_{n}=\tfrac{1}{n},\label{eq:frcd1}\\
\rho_{1}=\cdots=\rho_{n}.\label{eq:frcd2}
\end{gather}
\end{subequations}
\end{corollary}
%%%%%%%%%%%%%%%%%%%%%%
\begin{proof}
Assume that Eq.~\eqref{eq:fron} is satisfied. 
Condition~\eqref{eq:frcd1} is satisfied because $\{\eta_{i}\}_{i=1}^{n}$ is a probability distribution and 
\begin{equation}\label{eq:epln}
\eta_{i}\leqslant p_{\sf L}(\mathcal{E})
=\tfrac{1}{n}
\end{equation}
for all $i=1,\ldots,n$, where the first inequality is from Inequality~\eqref{eq:plet} and the last inequality is due to our assumption.
From Theorem~\ref{thm:plq1} together Eq.~\eqref{eq:fron} and Condition~\eqref{eq:frcd1}, we have
\begin{equation}\label{eq:rroi}
\tfrac{1}{n}\rho_{1}-\tfrac{1}{n}\rho_{i}\in\mathbb{SEP}^{*}
\end{equation}
for all $i=1,\ldots,n$. 
Inclusion~\eqref{eq:rroi} implies Condition~\eqref{eq:frcd2} since a block-positive operator is traceless if and only if it is the zero operator $\mathbb{O}$ in $\mathbb{H}$\cite{ha20225}.
%%%%%%%%%%%%%%%%%%%%%%

%%%%%%%%%%%%%%%%%%%%%%
Conversely, suppose that Condition~\eqref{eq:frcd} holds.
In this case, $\eta_{1}\rho_{1}-\eta_{i}\rho_{i}=\mathbb{O}$ for all $i=1,\ldots,n$.
Since $\mathbb{O}$ is obviously block positive, it follows from Theorem~\ref{thm:plq1} that Eq.~\eqref{eq:fron} is satisfied.
\end{proof}
%%%%%%%%%%%%%%%%%%%%%%

%%%%%%%%%%%%%%%%%%%%%%
%      Section       %
%%%%%%%%%%%%%%%%%%%%%%
\section{Quantum sequence discrimination}\label{sec:qseq}
%%%%%%%%%%%%%%%%%%%%%%
For a positive integer $L$, let us consider quantum state ensembles
\begin{equation}\label{eq:menb}
\mathcal{E}^{1}=\{\eta_{i}^{1},\rho_{i}^{1}\}_{i=1}^{n_{1}},\ldots,
\mathcal{E}^{L}=\{\eta_{i}^{L},\rho_{i}^{L}\}_{i=1}^{n_{L}},
\end{equation}
where each $\mathcal{E}^{l}$ consists of $n_{l}$ states $\rho_{1}^{l},\ldots,\rho_{n_{l}}^{l}$ with preparation probabilities $\eta_{1}^{l},\ldots,\eta_{n_{l}}^{l}$, for $l=1,\ldots,L$.
Here, we consider the situation that a sequence of quantum states is prepared from the ensembles in Eq.~\eqref{eq:menb};
initially, a state $\rho_{c_{1}}^{1}$ is prepared from the ensemble $\mathcal{E}^{1}$ with the corresponding probability $\eta_{c_{1}}^{1}$. Subsequently, at the $l$th instance for $l=2,\ldots,L$, a state $\rho_{c_{l}}^{l}$ is prepared from the ensemble $\mathcal{E}^{l}$ with the corresponding probability $\eta_{c_{l}}^{l}$.
As a result, a quantum sequence 
\begin{equation}\label{eq:qsdf}
(\rho_{c_{1}}^{1},\ldots,\rho_{c_{L}}^{L})
\end{equation}
is prepared with the probability 
\begin{equation}\label{eq:spdf}
\eta_{c_{1}}^{1}\times\cdots\times\eta_{c_{L}}^{L}.
\end{equation}
%%%%%%%%%%%%%%%%%%%%%%

%%%%%%%%%%%%%%%%%%%%%%
Let us denote $\mathbb{N}_{n_{l}}$ as the set of all integers from $1$ to $n_{l}$ for each $l=1,\ldots,L$. By using the notion
\begin{equation}\label{eq:vecn}
\vec{n}=(n_{1},\ldots,n_{L}),
\end{equation}
we denote $\mathbb{N}_{\vec{n}}$ as the Cartesian product of $\mathbb{N}_{n_{1}},\ldots,\mathbb{N}_{n_{L}}$, that is,
\begin{equation}\label{eq:cpzn}
\mathbb{N}_{\vec{n}}=\mathbb{N}_{n_{1}}\times\cdots\times\mathbb{N}_{n_{L}}.
\end{equation}
The quantum sequence in Eq.~\eqref{eq:qsdf} together with the probability in Eq.~\eqref{eq:spdf} can be considered as the tensor-producted state $\rho_{\vec{c}}$ with the probability $\eta_{\vec{c}}$,
\begin{equation}\label{eq:spvc}
\eta_{\vec{c}}=\prod_{l=1}^{L}\eta_{c_{l}}^{l},~~
\rho_{\vec{c}}=\bigotimes_{l=1}^{L}\rho_{c_{l}}^{l},
\end{equation}
for $\vec{c}=(c_{1},\ldots,c_{L})\in\mathbb{N}_{\vec{n}}$, from the tensor product of the ensembles in Eq.~\eqref{eq:menb},
\begin{equation}\label{eq:seqe}
\bigotimes_{l=1}^{L}\mathcal{E}^{l}=\{\eta_{\vec{c}},\rho_{\vec{c}}\}_{\vec{c}\in\mathbb{N}_{\vec{n}}}.
\end{equation}
%%%%%%%%%%%%%%%%%%%%%%

%%%%%%%%%%%%%%%%%%%%%%
In discriminating the quantum sequences from $\bigotimes_{l=1}^{L}\mathcal{E}^{l}$ in Eq.~\eqref{eq:seqe}, we use a measurement 
\begin{equation}\label{eq:dcrm}
\mathcal{M}=\{M_{\vec{c}}\}_{\vec{c}\in\mathbb{N}_{\vec{n}}},
\end{equation}
where each detection of $M_{\vec{c}}$ means that the prepared quantum sequence is guessed to be $\rho_{\vec{c}}$, for $\vec{c}\in\mathbb{N}_{\vec{n}}$.
In this case, the average probability of correctly guessing the prepared quantum sequence is
\begin{equation}\label{eq:pgle}
\sum_{\vec{c}\in\mathbb{N}_{\vec{n}}}\eta_{\vec{c}}\Tr(\rho_{\vec{c}}M_{\vec{c}}).
\end{equation}
The minimum-error discrimination of quantum sequence ensemble $\bigotimes_{l=1}^{L}\mathcal{E}^{l}$ is to maximize the average success probability in Eq.~\eqref{eq:pgle}, that is,
\begin{equation}\label{eq:pgsd}
p_{\sf G}\left(\bigotimes_{l=1}^{L}\mathcal{E}^{l}\right)=\max_{\mathcal{M}}\sum_{\vec{c}\in\mathbb{N}_{\vec{n}}}\eta_{\vec{c}}\Tr(\rho_{\vec{c}}M_{\vec{c}}),
\end{equation}
where the maximum is taken over all possible measurements $\mathcal{M}$ in Eq.~\eqref{eq:dcrm}.
%%%%%%%%%%%%%%%%%%%%%%

%%%%%%%%%%%%%%%%%%%%%%
Now, let us consider the situation that a measurement is performed at each step of quantum sequence; a measurement $\{M_{i}^{l}\}_{i=1}^{n_{l}}$ is performed at the $l$th step of the quantum sequence for $l=1,\ldots,L$. 
In this situation, the measurement operator $M_{\vec{c}}$ in Eq.~\eqref{eq:dcrm} can be represented as 
\begin{equation}\label{eq:mvcd}
M_{\vec{c}}=M_{c_{1}}^{1}\otimes\cdots\otimes M_{c_{L}}^{L}
\end{equation}
for each $\vec{c}=(c_{1},\ldots,c_{L})\in\mathbb{N}_{\vec{n}}$, therefore we have
\begin{equation}\label{eq:pgrw}
\sum_{\vec{c}\in\mathbb{N}_{\vec{n}}}\eta_{\vec{c}}\Tr(\rho_{\vec{c}}M_{\vec{c}})
=\prod_{l=1}^{L}\sum_{i=1}^{n_{l}}\eta_{i}^{l}\Tr(\rho_{i}^{l}M_{i}^{l}).
\end{equation}
In other words, Eq.~\eqref{eq:pgle} can be expressed as the product of average success probabilities for discriminating each ensemble $\mathcal{E}^{l}$ with respect to the measurement $\{M_{i}^{l}\}_{i=1}^{n_{l}}$, for $l=1,\ldots,L$.
%%%%%%%%%%%%%%%%%%%%%%

%%%%%%%%%%%%%%%%%%%%%%
Interestingly, it was recently shown that the maximum in Eq.~\eqref{eq:pgsd} can be achieved just by performing minimum-error discrimination independently at each step of the quantum sequence, that is,
\begin{equation}\label{eq:pgpd}
p_{\sf G}\left(\bigotimes_{l=1}^{L}\mathcal{E}^{l}\right)=\prod_{l=1}^{L}p_{\sf G}(\mathcal{E}^{l}).
\end{equation}
In other words, the minimum-error discrimination of $\bigotimes_{l=1}^{L}\mathcal{E}^{l}$ can always be factorized into that of each individual ensemble $\mathcal{E}^{l}$, for $l=1,\ldots,L$\cite{gupt20242}.
%%%%%%%%%%%%%%%%%%%%%%

%%%%%%%%%%%%%%%%%%%%%%
%      Section       %
%%%%%%%%%%%%%%%%%%%%%%
\section{LOCC discrimination of multi-party quantum sequences}\label{sec:ldmqs}
%%%%%%%%%%%%%%%%%%%%%%

%%%%%%%%%%%%%%%%%%%%%%
%    Subsection      %
%%%%%%%%%%%%%%%%%%%%%%
\subsection{Multi-party quantum sequence discrimination }\label{ssec:mqsq}
%%%%%%%%%%%%%%%%%%%%%%
In this section, we consider the case when the ensembles $\mathcal{E}^{1},\ldots,\mathcal{E}^{L}$ in Eq.~\eqref{eq:menb} consist of \emph{multi-party} quantum states;
each quantum state $\rho_{c_{l}}^{l}$ of the quantum sequence in Eq.~\eqref{eq:qsdf} is acting on a $m$-party Hilbert space, for $c_{l}=1,\ldots,n_{l}$ and $l=1,\ldots,L$.
Here, each $m$-party state $\rho_{c_{l}}^{l}$ of the quantum sequence is shared to $m$ parties $\mathsf{A}_{1},\ldots,\mathsf{A}_{m}$.
Figure~\ref{fig:febe} illustrates the process of preparing a quantum sequence from the $m$-party quantum sequence ensemble $\bigotimes_{l=1}^{L}\mathcal{E}^{l}$ in Eq.~\eqref{eq:seqe} and sharing it to the $m$ parties, $\mathsf{A}_{1},\ldots,\mathsf{A}_{m}$.
%%%%%%%%%%%%%%%%%%%%%%

%%%%%%%%%%%%%%%%%%%%%%
\begin{figure}[!tt]
\centerline{\includegraphics[scale=1.0]{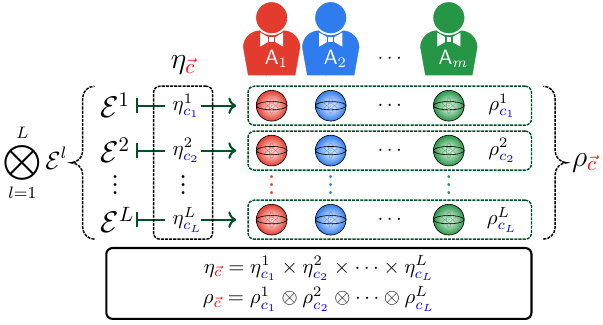}}
\caption{Multi-party quantum sequence ensemble $\bigotimes_{l=1}^{L}\mathcal{E}^{l}=\{\eta_{\vec{c}},\rho_{\vec{c}}\}_{\vec{c}\in\mathbb{N}_{\vec{n}}}$.
The $m$ parties, $\mathsf{A}_{1},\ldots,\mathsf{A}_{m}$, share a sequence of $L$ multi-party quantum states, $(\rho_{c_{1}}^{1},\ldots,\rho_{c_{L}}^{L})$. 
For each $l=1,\ldots,L$, the $l$th state $\rho_{c_{l}}^{l}$ is prepared from the ensemble $\mathcal{E}^{l}$ with the probability $\eta_{c_{l}}^{l}$, for $c_{l}=1,\ldots,n_{l}$.
The quantum sequence $(\rho_{c_{1}}^{1},\ldots,\rho_{c_{L}}^{L})$ can be regarded as the tensor producted state $\rho_{\vec{c}}$ with the probability $\eta_{\vec{c}}$, for $\vec{c}=(c_{1},\ldots,c_{L})\in\mathbb{N}_{\vec{n}}$. 
}\label{fig:febe}
\end{figure}
%%%%%%%%%%%%%%%%%%%%%%

%%%%%%%%%%%%%%%%%%%%%%
In this framework, we adopt the concepts of separability and block-positivity as defined in Eqs.~\eqref{eq:sep} and \eqref{eq:bpo}, with respect to $\mathsf{A}_{1},\ldots,\mathsf{A}_{m}$.
In other words, we use the sets $\mathbb{SEP}$ and $\mathbb{SEP}^{*}$ in Eqs.~\eqref{eq:sepd} and \eqref{eq:bpdf} to represent separability and block-positivity, respectively, with respect to $\mathsf{A}_{1},\ldots,\mathsf{A}_{m}$.
Accordingly, a measurement is classified as a separable measurement if each measurement operator is separable with respect to $\mathsf{A}_{1},\ldots,\mathsf{A}_{m}$.
Likewise, a measurement is considered as an LOCC measurement if it can be realized by LOCC among $\mathsf{A}_{1},\ldots,\mathsf{A}_{m}$.
%%%%%%%%%%%%%%%%%%%%%%

%%%%%%%%%%%%%%%%%%%%%%
In discriminating the quantum sequences from $\bigotimes_{l=1}^{L}\mathcal{E}^{l}$, if the measurements $\mathcal{M}$ in Eq.~\eqref{eq:dcrm} are limited to separable measurements, we denote the maximum of average success probability in Eq.~\eqref{eq:pgle} as
\begin{equation}\label{eq:mpsp}
p_{\sf SEP}\left(\bigotimes_{l=1}^{L}\mathcal{E}^{l}\right)=\max_{\textsf{Separable}\,\mathcal{M}}\sum_{\vec{c}\in\mathbb{N}_{\vec{n}}}\eta_{\vec{c}}\Tr(\rho_{\vec{c}}M_{\vec{c}}).
\end{equation}
Similarly, if the measurements $\mathcal{M}$ in Eq.~\eqref{eq:dcrm} are restricted to LOCC measurements, we denote
\begin{equation}\label{eq:mpel}
p_{\sf L}\left(\bigotimes_{l=1}^{L}\mathcal{E}^{l}\right)=\max_{\textsf{LOCC}\,\mathcal{M}}\sum_{\vec{c}\in\mathbb{N}_{\vec{n}}}\eta_{\vec{c}}\Tr(\rho_{\vec{c}}M_{\vec{c}}).
\end{equation}
%%%%%%%%%%%%%%%%%%%%%%

%%%%%%%%%%%%%%%%%%%%%%
Since performing a separable measurement independently at each step of the quantum sequence can be regarded as a separable measurement for the quantum sequence, it follows from Eq.~\eqref{eq:pgrw} together with the definitions in Eqs.~\eqref{eq:psep} and \eqref{eq:mpsp} that 
\begin{equation}\label{eq:pseq}
p_{\sf SEP}\left(\bigotimes_{l=1}^{L}\mathcal{E}^{l}\right)
\geqslant \prod_{l=1}^{L}p_{\sf SEP}(\mathcal{E}^{l}).
\end{equation}
Due to the similar reason, we also have
\begin{equation}\label{eq:pleq}
p_{\sf L}\left(\bigotimes_{l=1}^{L}\mathcal{E}^{l}\right)
\geqslant \prod_{l=1}^{L}p_{\sf L}(\mathcal{E}^{l}).
\end{equation}
%%%%%%%%%%%%%%%%%%%%%%

%%%%%%%%%%%%%%%%%%%%%%
Inequality~\eqref{eq:pleq} is saturated if and only if the maximum in Eq.~\eqref{eq:mpel} can be achieved just by performing optimal LOCC discrimination independently at each step of the quantum sequence. That is, the optimal LOCC discrimination of $\bigotimes_{l=1}^{L}\mathcal{E}^{l}$ can be factorized into that of each individual ensemble $\mathcal{E}^{l}$, for $l=1,\ldots,L$.
In the following definition, we introduce the concept of factorizability for the optimal LOCC discrimination of $\bigotimes_{l=1}^{L}\mathcal{E}^{l}$.
%%%%%%%%%%%%%%%%%%%%%%

%%%%%%%%%%%%%%%%%%%%%%
\begin{definition}\label{def:fact}
For a multi-party quantum sequence ensemble $\bigotimes_{l=1}^{L}\mathcal{E}^{l}$ in Eq.~\eqref{eq:seqe}, we say that the optimal LOCC discrimination of $\bigotimes_{l=1}^{L}\mathcal{E}^{l}$ is \emph{factorizable} if
\begin{equation}\label{eq:didp}
p_{\sf L}\left(\bigotimes_{l=1}^{L}\mathcal{E}^{l}\right)=\prod_{l=1}^{L}p_{\sf L}(\mathcal{E}^{l}).
\end{equation}
\end{definition}
%%%%%%%%%%%%%%%%%%%%%%

%%%%%%%%%%%%%%%%%%%%%%
%    Subsection      %
%%%%%%%%%%%%%%%%%%%%%%
\subsection{Factorizable LOCC discrimination of quantum sequences}\label{ssec:fldqs}
%%%%%%%%%%%%%%%%%%%%%%
In this subsection, we present our main results, providing conditions under which the optimal LOCC discrimination of a multi-party quantum sequence ensemble $\bigotimes_{l=1}^{L}\mathcal{E}^{l}$ in Eq.~\eqref{eq:seqe} becomes factorizable, that is, Eq.~\eqref{eq:didp} holds.
These conditions are established based on the following relation obtained from Inequalities~\eqref{eq:ineqp} and \eqref{eq:plet},
\begin{equation}\label{eq:grin}
\max_{\vec{c}\in\mathbb{N}_{\vec{n}}}\eta_{\vec{c}}\leqslant p_{\sf L}\left(\bigotimes_{l=1}^{L}\mathcal{E}^{l}\right)
\leqslant p_{\sf SEP}\left(\bigotimes_{l=1}^{L}\mathcal{E}^{l}\right)
\leqslant p_{\sf G}\left(\bigotimes_{l=1}^{L}\mathcal{E}^{l}\right).
\end{equation}
The first case we consider is when the first inequality in \eqref{eq:grin} is saturated. 
In this case, we show that Eq.~\eqref{eq:didp} holds.
Moreover, we further show that saturation of the first inequality naturally implies saturation of the second inequality in \eqref{eq:grin}.
%Moreover, the saturation of the second inequality in \eqref{eq:grin} is also guaranteed.
The second case we consider is when Eq.~\eqref{eq:didp} holds and both of the last two inequalities in \eqref{eq:grin} are saturated.
In this case, we show that the minimum-error discrimination of each individual ensemble $\mathcal{E}^{l}$ can be realized by LOCC, that is, $p_{\sf L}(\mathcal{E}^{l})=p_{\sf G}(\mathcal{E}^{l})$ for all $l=1,\ldots,L$.
We further show that the converse is also true.
Finally, we provide a necessary and sufficient condition satisfying both Eq.~\eqref{eq:didp} and saturation of the second inequality in \eqref{eq:grin}.
%%%%%%%%%%%%%%%%%%%%%%

%%%%%%%%%%%%%%%%%%%%%%
We start with the following theorem stating that if the first inequality in \eqref{eq:grin} is saturated, then the optimal LOCC discrimination of $\bigotimes_{l=1}^{L}\mathcal{E}^{l}$ is factorizable and the second inequality in \eqref{eq:grin} is also saturated.
%%%%%%%%%%%%%%%%%%%%%%
\begin{theorem}\label{thm:mpen}
For a multi-party quantum sequence ensemble $\bigotimes_{l=1}^{L}\mathcal{E}^{l}$ in Eq.~\eqref{eq:seqe} and $\vec{x}\in\mathbb{N}_{\vec{n}}$, we have 
\begin{equation}\label{eq:plon}
p_{\sf L}\left(\bigotimes_{l=1}^{L}\mathcal{E}^{l}\right)=\eta_{\vec{x}}
\end{equation}
if and only if
\begin{equation}\label{eq:nsbc}
\eta_{\vec{x}}\rho_{\vec{x}}-\eta_{\vec{c}}\rho_{\vec{c}}\in\mathbb{SEP}^{*}
\end{equation}
for all $\vec{c}\in\mathbb{N}_{\vec{n}}$.
In this case, we have
\begin{equation}\label{eq:else}
p_{\sf L}\left(\bigotimes_{l=1}^{L}\mathcal{E}^{l}\right)
=p_{\sf SEP}\left(\bigotimes_{l=1}^{L}\mathcal{E}^{l}\right)
\end{equation}
and the optimal LOCC discrimination of $\bigotimes_{l=1}^{L}\mathcal{E}^{l}$ is factorizable, that is, Eq.~\eqref{eq:didp} holds.
\end{theorem}
%%%%%%%%%%%%%%%%%%%%%%

%%%%%%%%%%%%%%%%%%%%%%
\begin{proof}
The first statement can be directly derived from Theorem~\ref{thm:plq1}.
The second statement is true due to Theorem~\ref{thm:plq1} together with
\begin{equation}\label{eq:corp}
\eta_{\vec{x}}\leqslant \prod_{l=1}^{L}p_{\sf L}(\mathcal{E}^{l})\leqslant p_{\sf L}\left(\bigotimes_{l=1}^{L}\mathcal{E}^{l}\right)=\eta_{\vec{x}},
\end{equation}
where the first two inequalities are from Inequalities~\eqref{eq:plet} and \eqref{eq:pleq}, respectively, and the last equality is by Eq.~\eqref{eq:plon}.
\end{proof}
%%%%%%%%%%%%%%%%%%%%%%

%%%%%%%%%%%%%%%%%%%%%%
Although Theorem~\ref{thm:mpen} provides a sufficient condition for the validity of Eq.~\eqref{eq:didp}, we note that Eq.~\eqref{eq:didp} does not generally hold, in contrast with Eq.~\eqref{eq:pgpd}.
That is, there exist ensembles $\mathcal{E}^{1},\ldots,\mathcal{E}^{L}$ violating Eq.~\eqref{eq:didp}, as illustrated in the following example.
%%%%%%%%%%%%%%%%%%%%%%

%%%%%%%%%%%%%%%%%%%%%%
\begin{example}\label{ex:pleo}
For integers $m,d,L\geqslant2$, let us consider the ensembles $\mathcal{E}^{1},\ldots,\mathcal{E}^{L}$, each identical to the $m$-qu$d$it state ensemble $\mathcal{E}=\{\eta_{1},\rho_{1};\eta_{2},\rho_{2}\}$ with
\begin{align}\label{eq:exer}
\eta_{1}=\tfrac{2d^{m}}{d+3d^{m}},~&\rho_{1}=\tfrac{1}{2d^{m}}(\mathbbm{1}_{d}^{m}+d^{m}\Phi_{d}^{m}),\nonumber\\
\eta_{2}=\tfrac{d+d^{m}}{d+3d^{m}},~&\rho_{2}=\Phi_{d}^{m}
\end{align}
where $\mathbbm{1}_{d}^{m}$ and $\Phi_{d}^{m}$ are the $m$-qu$d$it identity operator and \emph{Greenberger–Horne–Zeilinger}(GHZ) state, respectively\cite{gree1989},
\begin{eqnarray}
\mathbbm{1}_{d}^{m}&=&\sum_{i_{1},\ldots,i_{m}=0}^{d-1}\ket{i_{1}\cdots i_{m}}_{\mathsf{A}_{1}\cdots\mathsf{A}_{m}}\!\bra{i_{1}\cdots i_{m}},\nonumber\\
\Phi_{d}^{m}&=&\frac{1}{d}\sum_{i,j=0}^{d-1}\ket{i\cdots i}_{\mathsf{A}_{1}\cdots\mathsf{A}_{m}}\!\bra{j\cdots j}.\label{eq:maxe}
\end{eqnarray}
\end{example}
%%%%%%%%%%%%%%%%%%%%%%

%%%%%%%%%%%%%%%%%%%%%%
To show the invalidity of Eq.~\eqref{eq:didp} for the multi-party quantum sequence ensemble $\bigotimes_{l=1}^{L}\mathcal{E}^{l}$ in Example~\ref{ex:pleo}, we show that $\bigotimes_{l=1}^{L}\mathcal{E}^{l}$ satisfies
\begin{equation}\label{eq:sefc}
\prod_{l=1}^{L}p_{\sf L}(\mathcal{E}^{l})=\eta_{\vec{1}}
<p_{\sf L}\left(\bigotimes_{l=1}^{L}\mathcal{E}^{l}\right)
\end{equation}
where $\vec{1}\in\mathbb{N}_{\vec{n}}$ is the vector with all entries being $1$, that is,
\begin{equation}\label{eq:vone}
\vec{1}=(1,\ldots,1).
\end{equation}
We verify the equality and the inequality in Eq.~\eqref{eq:sefc} separately.
%%%%%%%%%%%%%%%%%%%%%%

%%%%%%%%%%%%%%%%%%%%%%
For the equality in Eq.~\eqref{eq:sefc}, we recall the block-positivity
\begin{equation}\label{eq:idps}
\mathbbm{1}_{d}^{m}-d\Phi_{d}^{m}\in\mathbb{SEP}^{*}
\end{equation}
for any $m,d\geqslant2$\cite{ha20231}.
Since
\begin{equation}\label{eq:exbl}
\eta_{1}\rho_{1}-\eta_{2}\rho_{2}=\tfrac{1}{d+3d^{m}}(\mathbbm{1}_{d}^{m}-d\Phi_{d}^{m})
\end{equation}
for the ensemble $\mathcal{E}$ in Eq.~\eqref{eq:exer}, it follows from Theorem~\ref{thm:plq1} and Inclusion~\eqref{eq:idps} that
\begin{equation}\label{eq:eplo}
p_{\sf L}(\mathcal{E})=\eta_{1}=\tfrac{2d^{m}}{d+3d^{m}}.
\end{equation}
Thus, Eq.~\eqref{eq:eplo} and the identical structure of $\mathcal{E}^{1},\ldots,\mathcal{E}^{L}$ in Example~\ref{ex:pleo} lead us to the equality in Eq.~\eqref{eq:sefc}.
%%%%%%%%%%%%%%%%%%%%%%

%%%%%%%%%%%%%%%%%%%%%%
For the inequality in Eq.~\eqref{eq:sefc}, we show that the quantum sequence ensemble $\bigotimes_{l=1}^{L}\mathcal{E}^{l}$ in Example~\ref{ex:pleo} does not satisfy 
\begin{equation}\label{eq:nbpp}
\eta_{\vec{1}}\rho_{\vec{1}}-\eta_{\vec{c}}\rho_{\vec{c}}\notin\mathbb{SEP}^{*}
\end{equation}
for the vector $\vec{c}\in\mathbb{N}_{\vec{n}}$ whose first entry is $2$ and the remaining entries are $1$, that is,
\begin{equation}\label{eq:vtdf}
\vec{c}=(2,1,\ldots,1).
\end{equation}
From Eqs.~\eqref{eq:exer} and \eqref{eq:exbl}, the left-hand side of Exclusion~\eqref{eq:nbpp} can be rewritten as
\begin{eqnarray}\label{eq:svdr}
\eta_{\vec{1}}\rho_{\vec{1}}-\eta_{\vec{c}}\rho_{\vec{c}}
&=&(\eta_{1}^{1}\rho_{1}^{1}-\eta_{2}^{1}\rho_{2}^{1})\otimes
\eta_{1}^{2}\rho_{1}^{2}\otimes\cdots\otimes\eta_{1}^{L}\rho_{1}^{L}\nonumber\\
&=&\tfrac{1}{(d+3d^{m})^{L}}(\mathbbm{1}_{d}^{m}-d\Phi_{d}^{m})\otimes(\mathbbm{1}_{d}^{m}+d^{m}\Phi_{d}^{m})\otimes\cdots\otimes(\mathbbm{1}_{d}^{m}+d^{m}\Phi_{d}^{m}).
\end{eqnarray}
Since the tensor product of a non-block-positive operator and a positive-semidefinite operator is not block positive, it follows from Eq.~\eqref{eq:svdr} that Exclusion~\eqref{eq:nbpp} is satisfied if 
\begin{equation}\label{eq:sbpp}
(\mathbbm{1}_{d}^{m}-d\Phi_{d}^{m})\otimes(\mathbbm{1}_{d}^{m}+d^{m}\Phi_{d}^{m})\notin\mathbb{SEP}^{*}.
\end{equation}
%%%%%%%%%%%%%%%%%%%%%%

%%%%%%%%%%%%%%%%%%%%%%
For the product state $\sigma$ which is the tensor product of $m$ copies of $2$-qu$d$it GHZ states $\Phi_{d}^{2}$, that is,
\begin{eqnarray}
\sigma&=&(\Phi_{d}^{2})_{\sf A_{1}}\otimes\cdots\otimes(\Phi_{d}^{2})_{\sf A_{m}}\nonumber\\
&=&\frac{1}{d^{m}}\sum_{i_{1},j_{1},\ldots,i_{m},j_{m}=0}^{d-1}
\ket{i_{1}\cdots i_{m}}_{\mathsf{A}_{1}\cdots\mathsf{A}_{m}}\!\bra{j_{1}\cdots j_{m}}\otimes
\ket{i_{1}\cdots i_{m}}_{\mathsf{A}_{1}\cdots\mathsf{A}_{m}}\!\bra{j_{1}\cdots j_{m}},
\label{eq:dses}
\end{eqnarray}
it is straightforward to verify
\begin{eqnarray}
\Tr[\sigma(\Phi_{d}^{m}\otimes\mathbbm{1}_{d}^{m})]
&=&\frac{1}{d}\sum_{i,j,k=0}^{d-1}
\Tr\Big[\sigma\Big(
\ket{i\cdots i}_{\mathsf{A}_{1}\cdots\mathsf{A}_{m}}\!\bra{j\cdots j}\otimes
\ket{k\cdots k}_{\mathsf{A}_{1}\cdots\mathsf{A}_{m}}\!\bra{k\cdots k}\Big)\Big]
=\frac{1}{d^{m}},\nonumber\\
\Tr[\sigma(\mathbbm{1}_{d}^{m}\otimes\Phi_{d}^{m})]
&=&\frac{1}{d}\sum_{i,j,k=0}^{d-1}
\Tr\Big[\sigma\Big(
\ket{k\cdots k}_{\mathsf{A}_{1}\cdots\mathsf{A}_{m}}\!\bra{k\cdots k}\otimes
\ket{i\cdots i}_{\mathsf{A}_{1}\cdots\mathsf{A}_{m}}\!\bra{j\cdots j}\Big)\Big]
=\frac{1}{d^{m}},\nonumber\\
\Tr[\sigma(\Phi_{d}^{m}\otimes\Phi_{d}^{m})]
&=&\frac{1}{d^{2}}\sum_{i,j,i',j'=0}^{d-1}
\Tr\Big[\sigma\Big(
\ket{i\cdots i}_{\mathsf{A}_{1}\cdots\mathsf{A}_{m}}\!\bra{j\cdots j}\otimes
\ket{i'\cdots i'}_{\mathsf{A}_{1}\cdots\mathsf{A}_{m}}\!\bra{j'\cdots j'}\Big)\Big]
=\frac{1}{d^{m}}.
\label{eq:sfvt}
\end{eqnarray}
From Eq.~\eqref{eq:sfvt}, we have
\begin{eqnarray}\label{eq:tvsv}
\Tr\big[\sigma \big((\mathbbm{1}_{d}^{m}-d\Phi_{d}^{m})\otimes(\mathbbm{1}_{d}^{m}+d^{m}\Phi_{d}^{m})\big)\big]
&=&1+\tfrac{1}{d^{m}}(d^{m}-d-d^{m+1}) \nonumber\\
&=&2-d-\tfrac{1}{d^{m-1}}<0.
\end{eqnarray}
Since Inequality~\eqref{eq:tvsv} implies Exclusion~\eqref{eq:sbpp}, Exclusion~\eqref{eq:nbpp} is satisfied. 
Thus, Exclusion~\eqref{eq:nbpp} and Theorem~\ref{thm:mpen} lead us to the inequality in Eq.~\eqref{eq:sefc}.
Now, Eq.~\eqref{eq:sefc} is satisfied for the quantum sequence ensemble $\bigotimes_{l=1}^{L}\mathcal{E}^{l}$ in Example~\ref{ex:pleo}, therefore Eq.~\eqref{eq:didp} does not hold.
In other words, the optimal LOCC discrimination of $\bigotimes_{l=1}^{L}\mathcal{E}^{l}$ is not factorizable.
%%%%%%%%%%%%%%%%%%%%%%

%%%%%%%%%%%%%%%%%%%%%%
We further remark that Eq.~\eqref{eq:plon} in Theorem~\ref{thm:mpen} guarantees 
\begin{equation}\label{eq:plof}
p_{\sf L}(\mathcal{E}^{l})=\eta_{x_{l}}^{l}
\end{equation}
for all $l=1,\ldots,L$, where $x_{l}$ is the $l$th entry of $\vec{x}=(x_{1},\ldots,x_{L})$, otherwise Inequality~\eqref{eq:corp} becomes strict.
On the other hand, we can see from Example~\ref{ex:pleo} that 
the converse is not true in general; Eq.~\eqref{eq:plon} is not necessarily guaranteed just by the validity of Eq.~\eqref{eq:plof} for all $l=1,\ldots,L$.
%%%%%%%%%%%%%%%%%%%%%%

%%%%%%%%%%%%%%%%%%%%%%
The following theorem shows that the optimal LOCC discrimination of $\bigotimes_{l=1}^{L}\mathcal{E}^{l}$ is factorizable and both of the last two inequalities in \eqref{eq:grin} are saturated if and only if the minimum-error discrimination of each individual ensemble $\mathcal{E}^{l}$ can be realized by LOCC, for $l=1,\ldots,L$.
%%%%%%%%%%%%%%%%%%%%%%
\begin{theorem}\label{thm:scfl}
For a multi-party quantum sequence ensemble $\bigotimes_{l=1}^{L}\mathcal{E}^{l}$ in Eq.~\eqref{eq:seqe}, we have
\begin{equation}\label{eq:plpg}
\prod_{l=1}^{L}p_{\sf L}(\mathcal{E}^{l})=p_{\sf L}\left(\bigotimes_{l=1}^{L}\mathcal{E}^{l}\right)=p_{\sf G}\left(\bigotimes_{l=1}^{L}\mathcal{E}^{l}\right)
\end{equation}
if and only if 
\begin{equation}\label{eq:lftc}
p_{\sf L}(\mathcal{E}^{l})=p_{\sf G}(\mathcal{E}^{l})
\end{equation}
for all $l=1,\ldots,L$.
\end{theorem}
%%%%%%%%%%%%%%%%%%%%%%
\begin{proof}
If Eqs.~\eqref{eq:didp} and \eqref{eq:plpg} are satisfied, then Eq.~\eqref{eq:lftc} holds for all $l=1,\ldots,L$, otherwise we have
\begin{equation}\label{eq:usow}
\prod_{l=1}^{L}p_{\sf L}(\mathcal{E}^{l})
< \prod_{l=1}^{L}p_{\sf G}(\mathcal{E}^{l})
=p_{\sf G}\left(\bigotimes_{l=1}^{L}\mathcal{E}^{l}\right)
=p_{\sf L}\left(\bigotimes_{l=1}^{L}\mathcal{E}^{l}\right),
\end{equation}
where the first equality is from Eq.~\eqref{eq:pgpd}.
The converse is also true because
\begin{equation}\label{eq:ptpi}
p_{\sf G}\left(\bigotimes_{l=1}^{L}\mathcal{E}^{l}\right)
=\prod_{l=1}^{L}p_{\sf L}(\mathcal{E}^{l})
\leqslant p_{\sf L}\left(\bigotimes_{l=1}^{L}\mathcal{E}^{l}\right)
\leqslant p_{\sf G}\left(\bigotimes_{l=1}^{L}\mathcal{E}^{l}\right),
\end{equation}
where the equality is due to Eqs.~\eqref{eq:pgpd} and \eqref{eq:lftc}, and
the inequalities is from Inequalities~\eqref{eq:pleq} and \eqref{eq:grin}.
\end{proof}
%%%%%%%%%%%%%%%%%%%%%%
 
%%%%%%%%%%%%%%%%%%%%%%
\begin{corollary}\label{cor:bepg}
Given that the optimal LOCC discrimination of $\bigotimes_{l=1}^{L}\mathcal{E}^{l}$ in Eq.~\eqref{eq:seqe} is factorizable, that is, Eq.~\eqref{eq:didp} holds, we have
\begin{equation}\label{eq:plbe}
\max_{\vec{c}\in\mathbb{N}_{\vec{n}}}\eta_{\vec{c}}<
p_{\sf L}\left(\bigotimes_{l=1}^{L}\mathcal{E}^{l}\right)
<p_{\sf G}\left(\bigotimes_{l=1}^{L}\mathcal{E}^{l}\right)
\end{equation}
if and only if
\begin{subequations}\label{eq:eane}
\begin{align}
\max\{\eta_{1}^{l},\ldots,\eta_{n_{l}}^{l}\}<p_{\sf L}(\mathcal{E}^{l}),\label{eq:eatc}
\\
p_{\sf L}(\mathcal{E}^{l'})<p_{\sf G}(\mathcal{E}^{l'})\label{eq:fsne}
\end{align}
\end{subequations}
for some $l,l'\in\{1,\ldots,L\}$.
\end{corollary}
%%%%%%%%%%%%%%%%%%%%%%
\begin{proof}
Let us first assume that the inequalities in \eqref{eq:plbe} are satisfied.
Under this assumption, Condition~\eqref{eq:eatc} holds; otherwise, Eq.~\eqref{eq:didp} leads us to saturation of the first inequality in \eqref{eq:grin}, which contradicts the first inequality in \eqref{eq:plbe}.
From Theorem~\ref{thm:scfl} and the second inequality in \eqref{eq:plbe}, it follows that Condition~\eqref{eq:fsne} is also satisfied.
%%%%%%%%%%%%%%%%%%%%%%

%%%%%%%%%%%%%%%%%%%%%%
Conversely, suppose that Condition~\eqref{eq:eane} holds.
The first inequality in \eqref{eq:plbe} then follows from Condition~\eqref{eq:eatc} and the argument in the paragraph containing Eq.~\eqref{eq:plof}.
The second inequality in \eqref{eq:plbe} also holds due to Eq.~\eqref{eq:didp} and Condition~\eqref{eq:fsne} together with Theorem~\ref{thm:scfl}.
\end{proof}
%%%%%%%%%%%%%%%%%%%%%%

%%%%%%%%%%%%%%%%%%%%%%
The following theorem establishes a necessary and sufficient condition satisfying both Eq.~\eqref{eq:didp} and saturation of the second inequality in \eqref{eq:grin}.
%%%%%%%%%%%%%%%%%%%%%%
\begin{theorem}\label{thm:sclf}
For a multi-party quantum sequence ensemble $\bigotimes_{l=1}^{L}\mathcal{E}^{l}$ in Eq.~\eqref{eq:seqe}, we have
\begin{equation}\label{eq:plsl}
\prod_{l=1}^{L}p_{\sf L}(\mathcal{E}^{l})=p_{\sf L}\left(\bigotimes_{l=1}^{L}\mathcal{E}^{l}\right)=p_{\sf SEP}\left(\bigotimes_{l=1}^{L}\mathcal{E}^{l}\right)
\end{equation}
if and only if there exist LOCC measurements $\{M_{i}^{1}\}_{i=1}^{n_{1}},\ldots,\{M_{i}^{L}\}_{i=1}^{n_{L}}$ and a Hermitian operator $H$ satisfying \begin{subequations}\label{eq:cfsp}
\begin{eqnarray}
H-\eta_{\vec{c}}\rho_{\vec{c}}\in\mathbb{SEP}^{*},\label{eq:cfsp1}\\
\Tr[M_{\vec{c}}(H-\eta_{\vec{c}}\rho_{\vec{c}})]=0\label{eq:cfsp2}
\end{eqnarray}
\end{subequations}
for all $\vec{c}=(c_{1},\ldots,c_{L})\in\mathbb{N}_{\vec{n}}$, where $M_{\vec{c}}$ is defined in Eq.~\eqref{eq:mvcd}.
\end{theorem}
%%%%%%%%%%%%%%%%%%%%%%
\begin{proof}
Let us assume that Eq.~\eqref{eq:plsl} is satisfied.
For each $l=1,\ldots,L$, we also denote $\{M_{i}^{l}\}_{i=1}^{n_{l}}$ as a LOCC measurement providing  $p_{\sf L}(\mathcal{E}^{l})$.
Since every LOCC measurement is a separable measurement, the LOCC measurement $\{M_{\vec{c}}\}_{\vec{c}\in\mathbb{N}_{\vec{n}}}$ in Eq.~\eqref{eq:mvcd} is a separable measurement and it provides
\begin{equation}\label{eq:lmsp}
\sum_{\vec{c}\in\mathbb{N}_{\vec{n}}}\eta_{\vec{c}}\Tr(\rho_{\vec{c}}M_{\vec{c}})
=\prod_{l=1}^{L}p_{\sf L}(\mathcal{E}^{l}),
\end{equation}
which are from Eq.~\eqref{eq:pgrw} and the assumption of $\{M_{i}^{l}\}_{i=1}^{n_{l}}$ for each $l=1,\ldots,L$. 
The right-hand side of Eq.~\eqref{eq:lmsp} can be rewritten as
\begin{equation}\label{eq:pppp}
\prod_{l=1}^{L}p_{\sf L}(\mathcal{E}^{l})=p_{\sf SEP}\left(\bigotimes_{l=1}^{L}\mathcal{E}^{l}\right)
\end{equation}
because
\begin{equation}\label{eq:mrim}
\prod_{l=1}^{L}p_{\sf L}(\mathcal{E}^{l})
\leqslant\prod_{l=1}^{L}p_{\sf SEP}(\mathcal{E}^{l})
\leqslant p_{\sf SEP}\left(\bigotimes_{l=1}^{L}\mathcal{E}^{l}\right)
=\prod_{l=1}^{L}p_{\sf L}(\mathcal{E}^{l}),
\end{equation}
where the inequalities are from Inequalities~\eqref{eq:spin} and \eqref{eq:pseq}, and the equality is due to Eq.~\eqref{eq:plsl}.
Thus, Proposition~\ref{pro:nsps} together with Eqs.~\eqref{eq:lmsp} and \eqref{eq:pppp} leads us to the existence of a Hermitian operator $H$ satisfying Condition~\eqref{eq:cfsp}.
%%%%%%%%%%%%%%%%%%%%%%

%%%%%%%%%%%%%%%%%%%%%%
Conversely, suppose that LOCC measurements $\{M_{i}^{1}\}_{i=1}^{n_{1}},\ldots,\{M_{i}^{L}\}_{i=1}^{n_{L}}$ and a Hermitian operator $H$ satisfy Condition~\eqref{eq:cfsp}.
Since every LOCC measurement is a separable measurement, the LOCC measurement $\{M_{\vec{c}}\}_{\vec{c}\in\mathbb{N}_{\vec{n}}}$ in Eq.~\eqref{eq:mvcd} is a separable measurement, therefore we have
\begin{equation}\label{eq:sppi}
p_{\sf SEP}\left(\bigotimes_{l=1}^{L}\mathcal{E}^{l}\right)
=\sum_{\vec{c}\in\mathbb{N}_{\vec{n}}}\eta_{\vec{c}}\Tr(\rho_{\vec{c}}M_{\vec{c}})
=\prod_{l=1}^{L}\sum_{i=1}^{n_{l}}\eta_{i}^{l}\Tr(\rho_{i}^{l}M_{i}^{l}),
\end{equation}
where the first and second equalities are from Proposition~\ref{pro:nsps} and Eq.~\eqref{eq:pgrw}, respectively.
The last term in Eq.~\eqref{eq:sppi} is bounded above as follows:
\begin{equation}\label{eq:dine}
\prod_{l=1}^{L}\sum_{i=1}^{n_{l}}\eta_{i}^{l}\Tr(\rho_{i}^{l}M_{i}^{l})
\leqslant\prod_{l=1}^{L}p_{\sf L}(\mathcal{E}^{l})
\leqslant
p_{\sf L}\left(\bigotimes_{l=1}^{L}\mathcal{E}^{l}\right)
\leqslant
p_{\sf SEP}\left(\bigotimes_{l=1}^{L}\mathcal{E}^{l}\right),
\end{equation}
where the first inequality is from the definition in Eq.~\eqref{eq:ple}, the second inequality is by
Inequalities~\eqref{eq:pleq}, and the last inequality is due to Inequality~\eqref{eq:grin}.
Thus, Eq.~\eqref{eq:sppi} and Inequality~\eqref{eq:dine} lead us to Eq.~\eqref{eq:plsl}.
\end{proof}
%%%%%%%%%%%%%%%%%%%%%%

%%%%%%%%%%%%%%%%%%%%%%
Unlike Theorems~\ref{thm:mpen} and \ref{thm:scfl} that cannot verify Eq.~\eqref{eq:didp} if Inequality~\eqref{eq:plbe} holds, we note that Theorem~\ref{thm:sclf} can be used to verify Eq.~\eqref{eq:didp} even when Inequality~\eqref{eq:plbe} holds.
We illustrate this in the following example.
%%%%%%%%%%%%%%%%%%%%%%

%%%%%%%%%%%%%%%%%%%%%%
\begin{example}\label{ex:maex}
For integers $m,d,L\geqslant2$, let us consider the ensembles $\mathcal{E}^{1},\ldots,\mathcal{E}^{L}$, each identical to the $m$-qu$d$it state ensemble $\mathcal{E}=\{\eta_{i},\rho_{i}\}_{i=1}^{d+2}$ consisting of $d+1$ separable states and one entangled state,
\begin{align}\label{eq:maex}
\eta_{i}=\frac{1}{d^{m}+d},~&\rho_{i}=\Psi_{i-1}^{m},~i=1,\ldots,d,\nonumber\\
\eta_{d+1}=\frac{d^{m}-d}{d^{m}+d},~&\rho_{d+1}=\frac{1}{d^{m}-d}\Big(\mathbbm{1}_{d}^{m}-\sum_{j=0}^{d-1}\Psi_{j}^{m}\Big),\nonumber\\
\eta_{d+2}=\frac{d}{d^{m}+d},~&\rho_{d+2}=\Phi_{d}^{m}
\end{align}
where $\mathbbm{1}_{d}^{m}$ and $\Phi_{d}^{m}$ are defined in Eq.~\eqref{eq:maxe} and 
\begin{equation}\label{eq:psii}
\Psi_{i}^{m}=\ket{i\cdots i}_{\mathsf{A}_{1}\cdots\mathsf{A}_{m}}\!\bra{i\cdots i}.
\end{equation}
\end{example}
%%%%%%%%%%%%%%%%%%%%%%

%%%%%%%%%%%%%%%%%%%%%%
To show the validity of Eq.~\eqref{eq:didp} for the quantum sequence ensemble $\bigotimes_{l=1}^{L}\mathcal{E}^{l}$ in Example~\ref{ex:maex}, we show that $\bigotimes_{l=1}^{L}\mathcal{E}^{l}$ satisfies Condition~\eqref{eq:cfsp} of Theorem~\ref{thm:sclf}.
For Condition~\eqref{eq:cfsp1}, let us consider the Hermitian operator
\begin{equation}\label{eq:ehde}
H=\tfrac{1}{(d^{m}+d)^{L}}\underbrace{\mathbbm{1}_{d}^{m}\otimes\cdots\otimes\mathbbm{1}_{d}^{m}}_{L}.
\end{equation}
Due to the identical structure of $\mathcal{E}^{1},\ldots,\mathcal{E}^{L}$ in Example~\ref{ex:maex} and the symmetry of $H$ in Eq.~\eqref{eq:ehde} under the permutation of $\mathcal{E}^{1},\ldots,\mathcal{E}^{L}$, Condition~\eqref{eq:cfsp1} holds for all $\vec{c}\in\mathbb{N}_{\vec{n}}$ if and only if it holds for all $\vec{c}=(c_{1},\ldots,c_{L})\in\mathbb{N}_{\vec{n}}$ with entries in non-increasing order, that is,
\begin{equation}\label{eq:dvdf}
c_{1}\geqslant\cdots\geqslant c_{L}.
\end{equation}
%%%%%%%%%%%%%%%%%%%%%%

%%%%%%%%%%%%%%%%%%%%%%
For each $\vec{c}=(c_{1},\ldots,c_{L})\in\mathbb{N}_{\vec{n}}$ with Eq.~\eqref{eq:dvdf}, we have
\begin{eqnarray}\label{eq:reex}
H-\eta_{\vec{c}}\rho_{\vec{c}}
&=&H-\underbrace{\eta_{d+2}\rho_{d+2}\otimes\cdots\otimes\eta_{d+2}\rho_{d+2}}_{t}\otimes
\eta_{c_{t+1}}\rho_{c_{t+1}}\otimes\cdots\otimes\eta_{c_{L}}\rho_{c_{L}}\nonumber\\
&=&\frac{1}{(d^{m}+d)^{L}}\Big(\mathbbm{1}_{d^{t}}^{m}\otimes\mathbbm{1}_{d}^{m}\otimes\cdots\otimes\mathbbm{1}_{d}^{m}-d^{t}\Phi_{d^{t}}^{m}\otimes R_{c_{t+1}}\otimes\cdots\otimes R_{c_{L}}\Big)
\end{eqnarray}
where $t\in\{0,1,\ldots,L\}$ denotes the number of entries equal to $d+2$ in $\vec{c}$, and $R_{1},\ldots,R_{d+2}$ are the operators obtained from  $\eta_{1}\rho_{1},\ldots,\eta_{d+2}\rho_{d+2}$ in Eq.~\eqref{eq:maex} by multiplying $d^{m}+d$, that is,
\begin{eqnarray}\label{eq:erdf}
R_{i}&=&\Psi_{i-1}^{m},~i=1,\ldots,d,\nonumber\\
R_{d+1}&=&\mathbbm{1}_{d}^{m}-\sum_{j=0}^{d-1}\Psi_{j}^{m},~R_{d+2}=d\Phi_{d}^{m}.
\end{eqnarray}
The last equality in Eq.~\eqref{eq:reex} holds because the tensor product of $t$ copies of $\mathbbm{1}_{d}^{m}$ is the $m$-qu$d^{t}$it identity operator $\mathbbm{1}_{d^{t}}^{m}$ and the tensor product of $t$ copies of $\Phi_{d}^{m}$ is the $m$-qu$d^{t}$it GHZ state $\Phi_{d^{t}}^{m}$.
%%%%%%%%%%%%%%%%%%%%%%

%%%%%%%%%%%%%%%%%%%%%%
From Eq.~\eqref{eq:reex} and the relation
\begin{equation}\label{eq:reab}
\bigotimes_{l=1}^{K}A_{l}-\bigotimes_{l=1}^{K}B_{l}=\sum_{k=1}^{K}A_{1}\otimes\cdots\otimes A_{K-k}\otimes (A_{K-k+1}-B_{K-k+1})\otimes B_{K-k+2}\otimes\cdots\otimes B_{K}
\end{equation}
for any operators $A_{1},B_{1},\ldots,A_{K},B_{K}$, we have
\begin{eqnarray}\label{eq:exrr}
H-\eta_{\vec{c}}\rho_{\vec{c}}
&=&\frac{1}{(d^{m}+d)^{L}}\Big(\mathbbm{1}_{d^{t}}^{m}-d^{t}\Phi_{d^{t}}^{m}\Big)\otimes R_{c_{t+1}}\otimes\cdots\otimes R_{c_{L}}
\nonumber\\
&&+\frac{1}{(d^{m}+d)^{L}}\sum_{l=t-1}^{L-2}\mathbbm{1}_{d^{t}}^{m}\otimes
\underbrace{\mathbbm{1}_{d}^{m}\otimes\cdots\otimes\mathbbm{1}_{d}^{m}}_{L-2-l}\otimes (\mathbbm{1}_{d}^{m}-R_{c_{t+L-2-l}})\otimes R_{c_{t+L-1-l}}\otimes\cdots\otimes R_{c_{L}}.
\end{eqnarray}
In Eq.~\eqref{eq:exrr}, $\mathbbm{1}_{d^{t}}^{m}-d^{t}\Phi_{d^{t}}^{m}$ is block positive by Inclusion~\eqref{eq:idps}, and both $R_{c_{l}}$ and $\mathbbm{1}_{d}^{m}-R_{c_{l}}$ are separable for all $l=t+1,\ldots,L$ because $R_{1},\ldots,R_{d+1}$ in Eq.~\eqref{eq:erdf} are separable.
Since the tensor product of a block-positive operator and a separable operators is block positive, the right-hand side of Eq.~\eqref{eq:exrr} is block positive\cite{rutk2014}.
Therefore, Condition~\eqref{eq:cfsp1} holds for all $\vec{c}\in\mathbb{N}_{\vec{n}}$.
%%%%%%%%%%%%%%%%%%%%%%

%%%%%%%%%%%%%%%%%%%%%%
For Condition~\eqref{eq:cfsp2}, consider the measurements $\{M_{i}^{1}\}_{i=1}^{d+2},\ldots,\{M_{i}^{L}\}_{i=1}^{d+2}$, each identical to the measurement $\mathcal{M}=\{M_{i}\}_{i=1}^{d+2}$,
\begin{eqnarray}\label{eq:emdd}
M_{i}&=&\Psi_{i-1}^{m},~i=1,\ldots,d,\nonumber\\
M_{d+1}&=&\mathbbm{1}_{d}^{m}-\sum_{j=0}^{d-1}\Psi_{j}^{m},~
M_{d+2}=\mathbb{O}_{d}^{m}
\end{eqnarray}
where $\mathbb{O}_{d}^{m}$ is the $m$-qu$d$it zero operator and $\Psi_{i}^{m}$ is defined in Eq.~\eqref{eq:psii}.
We also note that $\mathcal{M}$ in Eq.~\eqref{eq:emdd} is an LOCC measurement because it can be implemented by performing the same local measurement $\{\ket{i}\!\bra{i}\}_{i=0}^{d-1}$ on each party.
%%%%%%%%%%%%%%%%%%%%%%

%%%%%%%%%%%%%%%%%%%%%%
For each $\vec{c}=(c_{1},\ldots,c_{L})\in\mathbb{N}_{\vec{n}}$, we have
\begin{eqnarray}\label{eq:exex}
\Tr[M_{\vec{c}}(H-\eta_{\vec{c}}\rho_{\vec{c}})]
&=&\frac{1}{(d^{m}+d)^{L}}
\Tr\Bigg[M_{\vec{c}}
\sum_{l=1}^{L}\underbrace{\mathbbm{1}_{d}^{m}\otimes\cdots\otimes\mathbbm{1}_{d}^{m}}_{L-l}\otimes (\mathbbm{1}_{d}^{m}-R_{c_{L-l+1}})\otimes R_{c_{L-l+2}}\otimes\cdots\otimes R_{c_{L}}\Bigg]
\nonumber\\
&=&\frac{1}{(d^{m}+d)^{L}}
\sum_{l=1}^{L}
\Bigg[\prod_{k=1}^{L-l}\Tr M_{c_{k}}\Bigg]\Tr[M_{c_{L-l+1}}(\mathbbm{1}_{d}^{m}-R_{c_{L-l+1}})]
\Bigg[\prod_{k'=L-l+2}^{L}\Tr(M_{c_{k'}}R_{c_{k'}})\Bigg],
\end{eqnarray}
where $M_{\vec{c}}$ is defined in Eq.~\eqref{eq:mvcd}, $R_{i}$ is defined in Eq.~\eqref{eq:erdf}, and the first equality is from the relation in Eq.~\eqref{eq:reab}.
From the definitions of $R_{i}$ and $M_{i}$ in Eqs.~\eqref{eq:erdf} and \eqref{eq:emdd}, we can easily see that 
\begin{equation}\label{eq:timr}
\Tr[M_{i}(\mathbbm{1}-R_{i})]=0
\end{equation}
for all $i=1,\ldots,d+2$. 
Thus, Eqs.~\eqref{eq:exex} and \eqref{eq:timr} lead us to Condition~\eqref{eq:cfsp2}.
%%%%%%%%%%%%%%%%%%%%%%

%%%%%%%%%%%%%%%%%%%%%%
For the quantum sequence ensemble $\bigotimes_{l=1}^{L}\mathcal{E}^{l}$ in Example~\ref{ex:maex}, the Hermitian operator $H$ in Eq.~\eqref{eq:ehde} and the LOCC measurements $\{M_{i}^{1}\}_{i=1}^{d+2},\ldots,\{M_{i}^{L}\}_{i=1}^{d+2}$ in Eq.~\eqref{eq:emdd} satisfy Condition~\eqref{eq:cfsp} of Theorem~\ref{thm:sclf}, therefore Eq.~\eqref{eq:didp} holds.
In other words, the optimal LOCC discrimination of $\bigotimes_{l=1}^{L}\mathcal{E}^{l}$ is factorizable.
%%%%%%%%%%%%%%%%%%%%%%

%%%%%%%%%%%%%%%%%%%%%%
Now, we show the validity of Inequality~\eqref{eq:plbe} for the quantum sequence ensemble $\bigotimes_{l=1}^{L}\mathcal{E}^{l}$ in Example~\ref{ex:maex}. 
By using the result for the ensemble $\mathcal{E}$ in Eq.~\eqref{eq:maex}\cite{ha20231}, we have
\begin{equation}\label{eq:exre}
p_{\sf L}(\mathcal{E})=\tfrac{d^{m}}{d^{m}+d}<p_{\sf G}(\mathcal{E}).
\end{equation}
Moreover, Eq.~\eqref{eq:maex} leads us to
\begin{equation}\label{eq:gmpl}
\max\{\eta_{1},\ldots,\eta_{d+2}\}=\eta_{d+1}=\tfrac{d^{m}-d}{d^{m}+d}<
\tfrac{d^{m}}{d^{m}+d}
\end{equation}
for any $m,d\geqslant2$. 
From Eqs.~\eqref{eq:exre} and \eqref{eq:gmpl}, we have
\begin{equation}\label{eq:lbpg}
\max\{\eta_{1},\ldots,\eta_{d+2}\}<p_{\sf L}(\mathcal{E})<p_{\sf G}(\mathcal{E}).
\end{equation}
The quantum sequence ensemble $\bigotimes_{l=1}^{L}\mathcal{E}^{l}$ in Example~\ref{ex:maex} satisfies Eq.~\eqref{eq:didp} and Inequality~\eqref{eq:lbpg}.
Thus, Corollary~\ref{cor:bepg} implies Inequality~\eqref{eq:plbe} for the ensemble of the example.
%%%%%%%%%%%%%%%%%%%%%%

%%%%%%%%%%%%%%%%%%%%%%
%      Section       %
%%%%%%%%%%%%%%%%%%%%%%
\section{Discussion}\label{sec:dis}
%%%%%%%%%%%%%%%%%%%%%%
We have considered the discrimination of multi-party quantum sequences under LOCC constraints, and provided conditions under which the optimal LOCC discrimination of a multi-party quantum sequence ensemble can be factorized into that of each individual ensemble. 
We have further established a necessary and sufficient condition under which the optimal LOCC discrimination of a multi-party quantum state ensemble can be realized just by guessing the state with the largest probability. 
Our results have been illustrated with examples of multi-party quantum states in an arbitrary dimension, showing the cases that such factorizability of optimal LOCC discrimination is possible.
%%%%%%%%%%%%%%%%%%%%%%

%%%%%%%%%%%%%%%%%%%%%%
We note that our results offer a valuable framework for investigating the fundamental limits of quantum data hiding\cite{terh2001,divi2002,egge2002}. The concept of quantum data hiding is to conceal classical data from multiple players by using a multi-party quantum state ensemble, so that the data can be perfectly recovered by global measurements that require collaboration among all players, while LOCC measurements reveal no information about the data. 
%%%%%%%%%%%%%%%%%%%%%%

%%%%%%%%%%%%%%%%%%%%%%
Perfect recoverability requires the encoding states to be mutually orthogonal, whereas Corollary~\ref{cor:qdh} shows that complete LOCC-concealment requires the states to be mutually identical. 
Thus our result demonstrates that it is fundamentally impossible for any single quantum state ensemble to simultaneously guarantee perfect recoverability and complete concealment under LOCC operations.
Alternatively, various quantum data-hiding protocols have been developed that leverage quantum sequence discrimination to asymptotically limit the information accessible about the hidden data through LOCC measurements alone\cite{lami2018,lami2021,ha20241,ha20252}.
%%%%%%%%%%%%%%%%%%%%%%

%%%%%%%%%%%%%%%%%%%%%%
If we want to hide multiple classical bits by concealing each bit using a quantum sequence ensemble, the amount of information accessible about the data only by LOCC is determined by the optimal LOCC discrimination of the entire quantum sequence ensemble. 
%To guarantee that each bit is independently concealed, it is necessary for the optimal LOCC discrimination among the quantum sequence ensembles to be factorizable. 
To guarantee that each bit is independently concealed, it is necessary for the optimal LOCC discrimination of the entire quantum sequence ensemble to be factorized into that of each individual quantum sequence ensemble. 
Thus, our results in Theorems~\ref{thm:mpen}, \ref{thm:scfl} and \ref{thm:sclf} provide a foundation for the independent concealment of multiple classical bits using quantum sequences.
%%%%%%%%%%%%%%%%%%%%%%

%%%%%%%%%%%%%%%%%%%%%%
As illustrated in Example~\ref{ex:pleo}, the optimal LOCC discrimination of a quantum sequence ensemble is not always factorizable. 
The ensemble considered in this example is constructed as a tensor product of entangled state ensembles.
This naturally raises the question of whether the non-factorizability of optimal LOCC discrimination arises due to the presence of entanglement. 
Accordingly, an interesting direction for future research is to investigate whether there exists a separable quantum sequence ensemble for which the optimal LOCC discrimination is not factorizable.
%%%%%%%%%%%%%%%%%%%%%%

%%%%%%%%%%%%%%%%%%%%%%
%  Acknowledgments   %
%%%%%%%%%%%%%%%%%%%%%%
\section*{Acknowledgments}
This work was supported by Korea Research Institute for defense Technology planning and advancement (KRIT) grant funded by Defense Acquisition Program Administration(DAPA)(KRIT-CT-23–031), a National Research Foundation of Korea(NRF) grant funded by the Korean government(Ministry of Science and ICT) (No.NRF2023R1A2C1007039), and the Institute for Information \& Communications Technology Planning \& Evaluation(IITP) grant funded by the Korean government(MSIP)(Grant No. RS-2025-02304540). JSK was supported by Creation of the Quantum Information Science R\&D Ecosystem(Grant No. 2022M3H3A106307411) through the National Research Foundation of Korea(NRF) funded by the Korean government(Ministry of Science and ICT).
%%%%%%%%%%%%%%%%%%%%%%

%%%%%%%%%%%%%%%%%%%%%%
%     References     %
%%%%%%%%%%%%%%%%%%%%%%
 
%%%%%%%%%%%%%%%%%%%%%%
\end{document}